\documentclass[%
reprint,
aps,
mph,
superscriptaddress,
nofootinbib,
nobibnotes,
amsmath,amssymb,
longbibliography
]{revtex4-1}

\usepackage{xcolor}
\usepackage{tabularx}   
\usepackage{makecell}   
\usepackage{array} 
\usepackage{rotating}
\usepackage{graphicx}
\usepackage{dcolumn}
\usepackage{bm}

\newcolumntype{Y}{>{\centering\arraybackslash}X}
\newtheorem{theorem}{Theorem}[section]

\newtheorem{lemma}{Lemma}

\newenvironment{proof}{\par\noindent\textbf{Proof.}\ }{\hfill$\square$\par}

\begin{document}

\newcommand{\eg}{{\it e.g.}}
\newcommand{\ie}{{\it i.e.}}

\preprint{APS/123-QED}

\title{Sperner state and multipartite entanglement signals}

\author{Xin-Xiang Ju}
\email{juxinxiang21@mails.ucas.ac.cn}
 \affiliation{School of Physical Sciences, University of Chinese Academy of Sciences, Zhongguancun east road 80, Beijing 100190, China}

\author{Ya-Wen Sun}
 \email{yawen.sun@ucas.ac.cn}
 \affiliation{School of Physical Sciences, University of Chinese Academy of Sciences, Zhongguancun east road 80, Beijing 100190, China}%
\affiliation{Kavli Institute for Theoretical Sciences, University of Chinese Academy of Sciences, Beijing 100049, China}%

\author{Yang Zhao}
\email{zhaoyang20a@mails.ucas.ac.cn}
 \affiliation{School of Physical Sciences, University of Chinese Academy of Sciences, Zhongguancun east road 80, Beijing 100190, China}

\date{\today}

\begin{abstract} 

We establish a systematic classification scheme for multipartite entanglement structures. We define Sperner states\textemdash a broad class of states where apparent multipartite entanglement decomposes into fewer-partite entanglement among subsystems of each party. Each class of Sperner states is associated with one antichain hypergraph and each hypergraph encodes the maximal entanglement structure permissible under its constraints. We introduce a Multi-entanglement Measure Space (MEMS) where each Sperner class corresponds to a linear subspace defined by the vanishing of specific linear combinations of bipartite and multipartite measures. The nonvanishing of such combinations signals multipartite entanglement beyond the associated hypergraph, thereby distinguishing entanglement structures. We build a two way connection between each hypergraph entanglement structure and a distinct set of combinations, thereby quantifying the entanglement pattern and providing a unified basis for classifying all multipartite entanglement.

\end{abstract}

\maketitle

\section{Introduction}

Multipartite entanglement plays a central role across a wide range of modern physics, from quantum information and quantum many-body physics, to quantum gravity. Yet it remains notoriously difficult to define and diagnose\cite{Xie:2021hsy,Ma:2023ecg}. Although many multipartite entanglement measures have been proposed in recent years\cite{Umemoto:2018jpc,Bao:2018fso,Yang:2009ywa,Ju:2024hba,Ju:2025tgg}, including multi-entropy \cite{Gadde:2022cqi,Penington:2022dhr,Gadde:2023zzj}, a fundamental conceptual question persists: what concrete entanglement structures do these measures actually detect?

In this work, we bridge this gap by developing a comprehensive framework where linear combinations of multipartite entanglement measures act as precise detectors of entanglement structural patterns. We classify all 
$n$-party states by whether and how more-partite entanglement emerges from fewer-partite entanglement within smaller subsystems of each party, rather than as an irreducible global structure. Each such entanglement pattern can be mapped to a hypergraph satisfying an antichain condition, establishing what we term Sperner states. The key insight is that for each hypergraph, a characteristic set of linear combinations of  measures vanishes; conversely, non-vanishing values signal the presence of genuine multipartite entanglement
  beyond the structure captured by that hypergraph. This establishes a rigorous connection between linear combinations of multipartite entanglement measures and the precise entanglement structures they detect, providing a systematic framework for classifying the full landscape of multipartite entanglement.

\section{Sperner state and MEMS}\label{sec2}

{ {The construction of our general formalism for multipartite entanglement begins by classifying the states that \emph{lack} genuine irreducible multipartite entanglement. We define these as Sperner states, of which each class is associated with a hypergraph $H=(V,E)$, where 
$V=\{A,B,\dots,N\}$ denotes the set of $n$ physical subsystems (vertices), and 
$E\subseteq 2^V$ is a collection of hyperedges, with each $e\in E$ representing an arbitrary multipartite resource shared among the parties in $e$. The Sperner state class associated with $H$ is then defined as
\begin{equation}\label{def}
\rho_H=\Big(\otimes_{I\in V} U_{I}\Big)
\Big(\otimes_{e\in E}\rho_e\Big)
\Big(\otimes_{I\in V}U_{I}^\dagger\Big),
\end{equation}
where $U_I$ is an arbitrary local unitary acting on subsystem $I\in V$. Each $\rho_e$ is an arbitrary (pure or mixed) state supported on the subsystems in $e$ and may possess any form of multipartite entanglement. Before the unitary transformation, $\otimes_{e\in E}\rho_e$ is often considered as a freely available resource in state preparation \cite{Kraft:2021lzl,Hansenne:2023sze,Navascues:2020uwz,Luo:2020ugi}.
Crucially, when a subsystem $I$ participates in multiple hyperedges, its Hilbert space is understood to decompose into a tensor product of smaller factors $\mathcal{H}_I \cong \bigotimes_{e \ni I} \mathcal{H}_{I_e}$, with each factor appearing in exactly one $\rho_e$. 
In this way, different hyperedges act on disjoint parts of the same physical subsystem, and together prepare the full system. 

In this framework, any apparent $n$-partite entanglement is entirely reducible to the correlations shared within the smaller groups defined by $E$. Since multipartite entanglement is invariant under local unitaries, each hypergraph provides a faithful representation of the entanglement structures permitted in $\rho_H$: 
all entanglement is confined within hyperedges of this hypergraph, and no irreducible global $n$-partite resource is present beyond them. 

For instance, the ``triangle state"\cite{Zou:2020bly} $\rho_{\Delta}$ is associated with a hypergraph with edges $e_1=\{A,B\}$, $e_2=\{B,C\}$, and $e_3=\{C,A\}$. 
Its structure, \begin{scriptsize}
    \begin{equation}\label{eq:triangle-state}
\rho_{\Delta }
=(U_A\otimes U_B\otimes U_C)
\rho_{A_L B_R}\otimes\rho_{B_L C_R}\otimes\rho_{C_L A_R}(U_A^\dagger\otimes U_B^\dagger\otimes U_C^\dagger),
\end{equation}
\end{scriptsize}
explicitly shows how each subsystem is split to participate in separate entanglement patterns. {In the quantum network\cite{kimble2008quantum,ladd2010quantum} literature, this triangle state is preparable from the independent triangle network (ITN) \cite{Navascues:2020uwz}, and from the triangle network with local unitaries (UTN) \cite{Hansenne:2023sze}.  }

Since {each hyperedge represents the most general quantum state shared by the parties, }
adding an additional hyperedge supported on a \emph{strict subset} of an existing hyperedge is operationally redundant: it can always be absorbed into {the larger hyperedge}. 
This leads to our first (and only) structural rule:
\begin{quote}
\emph{No hyperedge may be contained in another.}
\end{quote}
Equivalently, $(E,\subseteq)$ forms an \emph{antichain}. 
This is precisely the same condition that appears in the classical Sperner theorem \cite{sperner1928satz}, which motivates our choice of the name, distinguishing it from other hypergraph-state constructions in the existing literature\cite{Xu:2023eyn,Rossi:2013rko,Bao:2020zgx,Hubeny:2024fjn}
\footnote{In quantum network studies, it is most common to apply local operations and shared randomness (LOSR)\cite{Navascues:2020uwz,Luo:2020ugi} on product states $\otimes_{e\in E}\rho_e$, whereas we consider local unitary (LU) transformation instead.}.

By design, Sperner states provide a rigorous taxonomy for multipartite entanglement and our central goal is therefore to distinguish these hypergraph multipartite structures, which requires appropriate measures. In this work, we move beyond the traditional debate over which specific multipartite measure is superior. Instead, we demonstrate that any sequence of bipartite and multipartite measures $E^{(m)} = E^{(m)}_{(A_1|\cdots|A_m)},\quad m=2,3,\dots,n,$ can be utilized, provided they satisfy a set of minimal consistency rules. Our key insight is that the specific multipartite entanglement structure of a hypergraph is not captured by a single measure, but is instead encoded in certain linear combinations of measures. These combinations serve as sharp ``structural signals," whose vanishing precisely signals the entanglement pattern associated with a given hypergraph, thereby distinguishing different multipartite entanglement structures in a unified way.

The four fundamental rules required on the measures $E^{(m)}$ are as follows.

\begin{itemize}
    \item \textbf{Additivity:} for product states $\rho\otimes\sigma$,
    \begin{equation}
       E^{(m)}(\rho\otimes\sigma)=E^{(m)}(\rho)+E^{(m)}(\sigma). 
    \end{equation}
    \item \textbf{LU invariance:} under local unitaries $U_{A_i}$,
    \begin{footnotesize} \begin{equation} 
        E^{(m)}(\rho)=E^{(m)}\bigl((U_{A_1}\otimes\cdots\otimes U_{A_m})\rho(U_{A_1}\otimes\cdots\otimes U_{A_m})^\dagger\bigr).
    \end{equation}
    \end{footnotesize}
    \item \textbf{Reducibility:} $E^{(1)}:=0$. If the global state factorizes as $\rho=\rho_{A_1\cdots A_{m-1}}\otimes \rho_{A_m}$, then
    \begin{equation}
        E^{(m)}_{(A_1|\cdots|A_{m-1}|A_m)}(\rho)=E^{(m-1)}_{(A_1|\cdots|A_{m-1})}(\rho_{A_1\cdots A_{m-1}}).
    \end{equation}
    \item \textbf{Symmetry:} $E^{(m)}$ is invariant under permutations of the parties,
    \begin{equation}
        E^{(m)}_{(\dots|A_i|\dots|A_j|\dots)}=E^{(m)}_{(\dots|A_j|\dots|A_i|\dots)}.
    \end{equation}
\end{itemize}
To our knowledge, essentially all commonly used bipartite entanglement measures and their multipartite generalizations satisfy these four properties. 
Examples include R\'enyi entropies and their multipartite generalizations  \cite{Gadde:2022cqi,Penington:2022dhr,Gadde:2023zzj}, entanglement of purification and its multipartite version \cite{Takayanagi:2017knl,Umemoto:2018jpc,Bao:2018fso}, squashed entanglement \cite{Christandl:2003med,Ju:2023dzo} and its multipartite extension \cite{Yang:2009ywa} etc.. 
Some measures are defined for pure states, while others extend naturally to mixed states. 
We will therefore select measures appropriate to the states under consideration, keeping our framework as general as possible.

The measures $\{E^{(m)}\}_{m=2}^n$ satisfying the four postulates could be applied to any nontrivial\footnote{Here nontrivial means that the partition in which all vertices are grouped into a single block is excluded.} partition $\ (A\dots|\dots|\dots N)$ of $V$. We denote the partition by $\pi$, and we can relabel $\{E^{(m)}\}_{m=2}^n$ by the corresponding $\pi$ and denote them as $E_\pi^{(|\pi|)}$.
However, individual measures are generally insufficient to isolate specific entanglement structures. Their true power emerges only when their collective behavior is considered. This motivates introducing the following space, namely the  \emph{Multi-entanglement Measure Space} (MEMS), which collects the full set of entanglement data.

\noindent\textbf{Definition 1 (Multi-entanglement Measure Space).}
The \emph{Multi-entanglement Measure Space} (MEMS) is the real vector space spanned by these coordinates,
\begin{equation}
    \mathrm{MEMS}_n := \mathrm{span}_{\mathbb{R}}\{E_\pi^{(|\pi|)}:\ \pi\in\Pi^\ast(V)\}.
\end{equation}
where $|\pi|$ is the number of partition blocks in $\pi$, and $\Pi^\ast(V)$ is the set of all nontrivial partitions of $V$.
Each quantum state $\rho$ on $V$ is {mapped to} a point in the MEMS
\begin{equation}
    \mathbf{E}(\rho):=\bigl(E_\pi^{(|\pi|)}(\rho)\bigr)_{\pi\in\Pi^\ast(V)}\in \mathrm{MEMS}_n .
\end{equation}

This definition naturally generalizes the usual entropy space \cite{zhang1997non,Pippenger:2003wky,Bao:2015bfa} by including multipartite as well as non-entropic quantities.
It is essential to include multipartite measures, since bipartite measures alone are generally insufficient to distinguish different Sperner classes, as we will explain in the next section.

{The dimensionality of $\text{MEMS}_n$ is determined by the number of distinct nontrivial partitions of the $n$-partite system. Each such partition gives a corresponding $E_\pi^{(|\pi|)}(\rho)$ that serves as a coordinate in this space. Since the total number of partitions of an $n$-element set is given by the Bell number $B_n$, removing the trivial partition reduces the dimension of the MEMS space to}
\begin{equation}
    \dim(\mathrm{MEMS}_n)=B_n-1.
\end{equation}

\section{Geometric correspondence of a Sperner state class in \textbf{MEMS}}

{We want to map each hypergraph into a geometric object in MEMS, so that from the position in MEMS, we can tell the entanglement structure of a given state. The question is therefore}: \emph{What geometric object $G_H$ do the Sperner states of a given hypergraph} $H$ correspond to in MEMS? In other words, {what relations among the coordinates $\{E_\pi^{(|\pi|)}\}_{\pi\in\Pi^\ast(V)}$ must hold for states of a given Sperner hypergraph?}
We start from two trivial examples:
\begin{enumerate}
    \item A hypergraph consisting of a single hyperedge $e=V$ (i.e., $|e|=n$), which correlates all parties.
    \item A hypergraph with no hyperedges.
\end{enumerate}
The first case corresponds to a completely general $n$-party state, hence it can in principle realize any point in MEMS.\footnote{We do not assume additional properties of $E_\pi^{(|\pi|)}$ such as positivity, monogamy, or entropy inequalities. Therefore, unlike the entropy cone\cite{zhang1997non,Pippenger:2003wky,Bao:2015bfa}, MEMS does not carry extra convex-cone structure by default.}
The second case contains no entangling operation of any kind; by additivity and reducibility, every $E_\pi^{(|\pi|)}$ must vanish, so this state corresponds to the origin $\mathbf{0}\in \mathrm{MEMS}_n$.

These examples suggest that some Sperner states are ``more general'' than others. 
More precisely, given two hypergraphs on the same vertex set $V$, say $H_1=(V,E_1)$ and $H_2=(V,E_2)$, 
if for any hyperedge $e_{i_2}$ in $H_2$, one can always find another hyperedge  $e_{i_1}\supset e_{i_2}$ in $H_2$, we write $H_1\succeq H_2$ (equivalently, $E_2$ is ``covered'' by $E_1$), meaning that the Sperner states associated with $H_1$ are more general and allow more kinds of entanglement structures.
Correspondingly, in MEMS, the geometric objects should satisfy the same containment relation: $G_{H_2}$ should be contained in $G_{H_1}$.

For other classes of Sperner states, a specific family of linear equalities can be identified.
Let us take the triangle state $\rho_{\Delta ABC}$ in \eqref{eq:triangle-state} as an example. 
By the \emph{reducibility} and \emph{additivity} properties of the multipartite measures, we obtain
\begin{scriptsize}
\begin{equation}\label{eq:triangle-macro-micro}
\begin{aligned}
E^{(2)}_{(A|BC)}\big|_{\rho_{\Delta}}
&=E^{(2)}\big|_{\rho_{A_LB_R}}+E^{(2)}\big|_{\rho_{C_LA_R}},\\
E^{(2)}_{(B|CA)}\big|_{\rho_{\Delta }}
&=E^{(2)}\big|_{\rho_{B_LC_R}}+E^{(2)}\big|_{\rho_{A_LB_R}},\\
E^{(2)}_{(C|AB)}\big|_{\rho_{\Delta }}
&=E^{(2)}\big|_{\rho_{C_LA_R}}+E^{(2)}\big|_{\rho_{B_LC_R}},\\
E^{(3)}_{(A|B|C)}\big|_{\rho_{\Delta }}
&=E^{(2)}\big|_{\rho_{A_LB_R}}
+E^{(2)}\big|_{\rho_{B_LC_R}}
+E^{(2)}\big|_{\rho_{C_LA_R}}.
\end{aligned}
\end{equation}
\end{scriptsize}
It then follows immediately that
\begin{tiny}
\begin{equation}\label{eq:triangle-linear-relation}
E^{(3)}_{(A|B|C)}\big|_{\rho_{\Delta}}
-\frac12\left(
E^{(2)}_{(A|BC)}\big|_{\rho_{\Delta}}
+E^{(2)}_{(B|AC)}\big|_{\rho_{\Delta}}
+E^{(2)}_{(C|AB)}\big|_{\rho_{\Delta}}
\right)=0.
\end{equation}
\end{tiny}

{As illustrated by the triangle state, entanglement structures are determined by the correlations among subsystems of the $n$ parties (e.g., $A_L, A_R$). Since physical observables must be those of each whole party due to LU invariance, these internal relations must be embodied as constraints on global measures, which manifest as vanishing conditions of specific linear combinations.
}

We now generalize this procedure to an arbitrary class of Sperner states associated with a Sperner hypergraph
$H=(V,E)$ where $|V|=n$. 
For clarity, we distinguish two types of quantities:
\begin{itemize}
    \item \textbf{Macroscopic quantities:} the measures $E_\pi^{(|\pi|)}(\rho)$ evaluated on the \emph{entire} $n$-party state {treating each party as a whole},
    where $\pi\in\Pi^\ast(V)$ labels a nontrivial set partition of $V$. {These quantities span the MEMS. With $n$ fixed, macroscopic quantities are the same for all hypergraphs.}
    \item \textbf{Microscopic quantities:} the measures evaluated on the factor states {$\rho_e$} in (\ref{def}). 
    Concretely, for each hyperedge $e\in E$ and each $\pi_e\in\Pi^\ast(e)$, we include the quantity
    $E_{\pi_e}^{(|\pi_e|)}(\rho_e)$, where $\rho_e$ denotes the (generally mixed) state supported on the subsystems in $e$. {These quantities could involve subsystems of each party, and the number of them depends on the structures of specific hypergraphs.}
\end{itemize}

{In summary, macroscopic quantities are for $n$ parties e.g., $A$, $B$.., while microscopic quantities are for subsystems of each party, e.g., $A_L$, $B_R$, etc..} Each microscopic quantity is indexed by a pair $(e,\pi_e)$ with $e\in E$ and $\pi_e\in\Pi^\ast(e)$.
Hence the total number of microscopic coordinates is $\sum_{e\in E}\bigl(B_{|e|}-1\bigr),$
and, because each $\rho_e$ an arbitrary state, we treat these microscopic quantities as algebraically independent parameters.

Define the \emph{macroscopic vector} and the \emph{microscopic vector} by
\begin{align}
   &\mathbf{V}^{\mathrm{macro}}
:=
\bigl(E_\pi^{(|\pi|)}(\rho)\bigr)_{\pi\in\Pi^\ast(V)},\notag\\
&\mathbf{V}^{\mathrm{micro}}
:=
\bigl(E_{\pi_e}^{(|\pi_e|)}(\rho_e)\bigr)_{(e,\pi_e):\ e\in E,\ \pi_e\in\Pi^\ast(e)}.
\end{align}
Then, equation \eqref{eq:triangle-macro-micro} is an instance of a general linear relation
\begin{equation}\label{eq:R-mapping}
\mathbf{V}^{\mathrm{macro}} = R(H)\mathbf{V}^{\mathrm{micro}},
\end{equation}
where $R(H)$ is a $0$--$1$ valued matrix determined entirely by the hypergraph structure and the reduction rule below.

\noindent\textbf{Definition 2: Partition-reduction Matrix $R(H)$.} 

Fix a macroscopic partition $\pi\in\Pi^\ast(V)$ and a microscopic label $(e,\pi_e)$ with $e\in E$ and $\pi_e\in\Pi^\ast(e)$.
Let $\pi|_e\in\Pi(e)$ denote the restriction of $\pi$ to $e$, obtained by taking the intersection of each block of $\pi$ with $e$
and deleting empty intersections.
Then we define
\begin{equation}\label{eq:R-entry}
R(H)_{\pi,(e,\pi_e)}
:=
\begin{cases}
1, & \text{if }\pi|_e=\pi_e,\\
0, & \text{otherwise}.
\end{cases}
\end{equation}
Equivalently, $R(H)_{\pi,(e,\pi_e)}=1$ if and only if the multipartition $\pi$ on $V$,
when reduced to the subsystems inside $e$ (by eliminating all parties outside)
matches exactly the microscopic multipartition $\pi_e$.

Given the linear map \eqref{eq:R-mapping}, all linear equalities among the macroscopic coordinates
$\mathbf{V}^{\mathrm{macro}}$ follow from elementary linear algebra: they are precisely the linear relations in the left nullspace
of $R(H)$.
In Appendix \ref{appD} we list several explicit hypergraph examples, together with their corresponding matrices $R(H)$
and the complete set of linear equalities satisfied by $\mathbf{V}^{\mathrm{macro}}$.

In $\mathrm{MEMS}_n$, each linear equality among the coordinates $\{E_\pi^{(|\pi|)}\}_{\pi\in\Pi^\ast(V)}$ defines a codimension-one linear hyperplane passing through the origin.\footnote{This geometric viewpoint is exactly analogous to that of the entropy space (and entropy cone) literature \cite{zhang1997non,Pippenger:2003wky,Bao:2015bfa}.} 
For a fixed Sperner hypergraph $H=(V,E)$, the corresponding class of Sperner states is therefore represented by a linear subspace $G_H$ obtained as the intersection of all such hyperplanes that the states must reside in. 
The dimension of this object is an important quantity: it measures how general the Sperner class is.
By definition, the dimension of $G_H$ equals the number of linearly independent macroscopic quantities compatible with $H$, given by the rank of its partition-reduction matrix $R(H)$.

However, evaluating $\operatorname{rank}_{\mathbb{R}}\bigl(R(H)\bigr) $ by explicitly constructing the matrix $R(H)$ is computationally expensive when $n$ is large, and it also obscures the goal of reading the entanglement structure directly from the hypergraph. 
We therefore seek a closed-form expression for $\operatorname{rank}_{\mathbb{R}}\bigl(R(H)\bigr) $ in terms of the overlap structure of hyperedges.

\begin{theorem}\label{thm1}
The rank of the partition-reduction matrix equals the following inclusion-exclusion quantity:
\begin{equation}
    \operatorname{rank}_{\mathbb{R}}\bigl(R(H)\bigr)= \sum_{\emptyset\neq F\subseteq E} (-1)^{|F|+1}\bigl(B_{k(F)}-1\bigr), 
\end{equation} where $k(F):=|\bigcap_{e_i\in F}e_i|$ counts the number of vertices in the common intersection of the hyperedges in $F$. 
\end{theorem}
\textbf{Proof:} See Appendix \ref{appA}.

Naively, one might expect that each microscopic term arising from a hyperedge should provide an independent free parameter of the Sperner states, thereby contributing directly to the dimension of the corresponding geometric object in $\mathrm{MEMS}_n$. 
However, this naive expectation fails whenever two or more hyperedges overlap: overlaps create linear dependencies among the macroscopic coordinates induced by different hyperedges. 
The correct counting is governed by an inclusion-exclusion principle, which connects the overlap pattern of hyperedges directly to the generality (codimension) of the Sperner states.

A second foundational question is whether the map from a hypergraph to its corresponding geometric object $H\mapsto G_H$ is injective. In other words, we do not want two different Sperner hypergraphs to impose exactly the same set of equalities; otherwise, $\mathrm{MEMS}_n$ would fail to distinguish different multipartite entanglement patterns. 
Fortunately, the answer is affirmative. 

\begin{theorem}\label{thm2}
(uniqueness of equalities) There do not exist two different Sperner hypergraphs that share exactly the same set of linear equalities in $\mathrm{MEMS}_n$.
\end{theorem}
\textbf{Proof:} See Appendix \ref{appB}.

This viewpoint also clarifies why multipartite entanglement measures are necessary, rather than relying only on bipartite ones. 
Bipartite measures alone cannot distinguish all different classes of Sperner states simultaneously: the simplest example is that the triangle state \eqref{eq:triangle-state} can share the same bipartite data as a generic tripartite state. However, the full set of multipartite coordinates in $\mathrm{MEMS}_n$ can separate them.

Having established a correspondence between geometric objects in $\mathrm{MEMS}_n$ and Sperner states, we next translate basic geometric operations in $\mathrm{MEMS}_n$ into operations on Sperner hypergraphs. 
Two fundamental operations on subspaces in a linear space are: (i) forming the \emph{linear span} of two subsets (the smallest linear subspace containing both of them), and (ii) taking their \emph{intersection}. 
We find that these operations admit natural counterparts in the hypergraph language.

\begin{theorem}\label{thm3}
Let $G_{H_a},G_{H_b}\subseteq \mathrm{MEMS}_n$ be the geometric objects corresponding to two Sperner hypergraphs
$H_a=(V,E_a)$ and $H_b=(V,E_b)$. 
Then the linear span $\mathrm{span}(G_{H_a}\cup G_{H_b})$ corresponds to the Sperner hypergraph
\begin{equation*}
    H_{a\vee b}:=\mathrm{Antichain}(E_a\cup E_b),
\end{equation*}
i.e., the hypergraph obtained by taking the union of all hyperedges in $H_a$ and $H_b$ and then removing redundant
hyperedges so that the resulting edge set is an antichain.
\end{theorem}

\begin{theorem}\label{thm4}
Let $G_{H_a},G_{H_b}\subseteq \mathrm{MEMS}_n$ be the geometric objects defined above. 
Then the intersection $G_{H_a}\cap G_{H_b}$ corresponds to the Sperner hypergraph
\begin{equation*}
    H_{a\wedge b}:=\mathrm{Antichain}\bigl(\{\, e\cap f\ :\ e\in E_a,\ f\in E_b,\ |e\cap f|\ge 2\,\}\bigr),
\end{equation*}
i.e., the hypergraph whose hyperedges are all nontrivial pairwise intersections of the hyperedges of $H_a$ and $H_b$,
followed by the removal of redundancies to restore the antichain property.
\end{theorem}
 
Proofs of Theorems \ref{thm3} and \ref{thm4} are given in Appendix \ref{appC}.

\section{Generating multipartite entanglement signals }
 
Since a geometric object in $\mathrm{MEMS}_n$ is equivalently characterized by the set of linear equalities it obeys,
we may rephrase Theorems \ref{thm3}-\ref{thm4}  directly in the language of equalities:
\begin{itemize}
    \item The set of equalities shared by two classes of Sperner states is precisely the set of equalities obeyed by their \emph{union} class of Sperner states associated with  $H_{A\vee B}$.
    \item The set of equalities that appear in \emph{either} of two classes of Sperner states is precisely the set of equalities obeyed by their \emph{intersection} class of Sperner states associated with $H_{A\wedge B}$.
\end{itemize}
Using these, we can assign an explicit information-theoretic interpretation to arbitrary linear combinations of the measures $\{E_\pi^{(|\pi|)}\}$ {by finding their corresponding Sperner states via procedures as follows.}
 
Given a collection of linear equalities $C_j(\{E_\pi^{(|\pi|)}\})=0$, $(j=1,2,\dots)$,
where each $C_j$ is a fixed linear functional of the coordinates $\{E_\pi^{(|\pi|)}\}_{\pi\in\Pi^\ast(V)}$, for each constraint $C_j(\{E_\pi^{(|\pi|)}\})=0$, we can search all Sperner states that satisfy this equality. Taking the union over the hyperedges of all their associated hypergraphs, we will get the unique, most general hypergraph $H_{j}$, by Theorem \ref{thm3} (the union operation).

By construction, the constraint $C_j(\{E_\pi^{(|\pi|)}\})=0$ vanishes precisely for the multipartite entanglement structures present in $H_j$, and is therefore \emph{insensitive}. Equivalently, the functional $C_j$ can only ``detect'' entanglement patterns beyond those encoded by $H_{j}$.

Finally, by applying Theorem \ref{thm4} (the intersection operation), we may impose all constraints simultaneously: Taking
$H_\ast :=\bigwedge_j H_{j}$
produces a unique (typically less general) Sperner hypergraph $H_\ast$ that satisfies \emph{all} equalities
$C_j(\{E_\pi^{(|\pi|)}\})=0$ for $j=1,2,\dots$.
In other words, the entire family of constraints $\{C_j=0\}$ is \emph{totally insensitive} to the entanglement
structure encoded by $H_\ast$, and therefore must detect any entanglement structure that lies
beyond $H_\ast$.

{As an application, we construct signals which are sensitive to $k$-partite entanglement in an $n$-partite system.}

\noindent\textbf{Definition 3 ($k$-simple product state).}
An $n$-partite state $\rho$ on parties $V$ is called a \emph{$k$-simple product state} if there exists a partition
\begin{equation}
    V = S_1 \sqcup \cdots \sqcup S_r,
\quad
\max_{a\in[r]} |S_a| = k,
\end{equation}
such that the global state factorizes as
\begin{equation}
    \rho = \rho_{S_1}\otimes\cdots\otimes \rho_{S_r},
\end{equation}
where $\rho$ is a tensor product across a grouping of the parties, and the largest group contains $k$ parties.

The equalities obeyed by the $k$-uniform-complete hypergraph $H^{(k)}_{\mathrm{comp}}$\footnote{{A $k$-uniform-complete hypergraph is defined as a hypergraph in which every hyperedge has size $k$, and every $k$-element subset of $V$ forms a hyperedge.}} identify precisely with the linear combinations of the measures $\{E_\pi^{(|\pi|)}\}$ that are \emph{insensitive} to entanglement involving at most $k$ parties, but can remain sensitive to entanglement involving more than $k$ parties. 
In particular, these combinations vanish for every $k$-simple product state. 
Moreover, by Theorem \ref{thm3}, these are the \emph{only} linear combinations that vanish for \emph{all} $k$-simple product states.

The number of linearly independent combinations that are \emph{sensitive} to $k$-partite entanglement but \emph{insensitive} to entanglement involving fewer parties is given by the difference between the number of equalities obeyed by the $(k-1)$-complete hypergraph and the number of equalities obeyed by the $k$-complete hypergraph. By Theorem 1, this number is
$\binom{n}{k}\sum_{j=0}^{k}(-1)^j\binom{k}{j}\,B_{k-j}.$ We provide n=4 examples in Appendix \ref{appD}.

It is worth noting that some of these linear combinations have appeared previously in special settings. 
For example, in the case $n=3$ and $k=3$, there is only one independent combination, namely the relation in \eqref{eq:triangle-linear-relation}. 
Ref.~\cite{Liu:2024ulq,Harper:2025uui} studied {this case} when $E^{(m)}$ are taken to be multi-entropies, while Ref.~\cite{Bao:2025psl} studied the case when $E^{(m)}$ are taken to be multipartite EWCS quantities in holography\cite{Maldacena:1997re,Takayanagi:2017knl}. 
For the case $n=k$ with $E^{(m)}$ given by R\'enyi multi-entropies, Ref.~\cite{Iizuka:2025ioc} studied additional constraints obtained by imposing the permutation symmetry among the $n$ subsystems. 
We believe that our framework unifies and generalizes these results by organizing them systematically in terms of Sperner hypergraphs and their corresponding geometric objects in $\mathrm{MEMS}_n$.

\section*{Acknowledgement}

We thank Bart{\l}omiej Czech for useful discussion.
This work was supported by Project 12575068 supported by the National Natural Science Foundation of China.
\bibliography{apssamp}

\newpage
\onecolumngrid
\appendix
\section{The number of equalities obeyed by a Sperner class}\label{appA}

In this appendix, we prove Theorem \ref{thm1}. By calculating the rank of the partition-reduction matrix, Theorem \ref{thm1} provides an explicit expression of the number of equalities that a class of Sperner states satisfies.

First, we need to view the column space of the partition-reduction matrix as a particular function space, so that the underlying combinatorial structures become manifest, and the rank (dimension) calculation is more straightforward.
We begin by defining, for a vertex subset $S\subseteq V,\ |S|\ge 2$ and a specific partition $\pi\in\Pi(S)$\footnote{$\Pi(S)$ is the collection of all partitions on $S$, including the trivial one.}, the indicator function $\mathbf{1}_{S,\pi}$ acting on nontrivial partition of vertices by
\begin{equation}
\mathbf{1}_{S,\pi}(P)=
\begin{cases}
1,&\text{if } P|_{S}=\pi,\\
0,&\text{otherwise},
\end{cases}
\quad P\in \Pi^\ast(V).
\end{equation}
Thus, every column of the partition-reduction matrix $R(H)$ indexed by hyperedge $e_i$ and partition $\sigma\in \Pi^\ast(e_i)$ can be naturally identified with the function $\mathbf{1}_{e_i,\sigma}$.
Furthermore, we can associate to any $S$ the subspace spanned by all the indicator functions $\mathbf{1}_{S,\pi}$ with $\pi$ a nontrivial partition on $S$\footnote{Unlike $\mathbf{1}_{S,\pi}$, we can also define $U_S$ when $|S|\le1$. It is easy to see that $U_S=\{0\}$ in this case.}:
\begin{equation}
U_{S}:=\operatorname{span}\{\mathbf{1}_{S,\pi}:\ \pi\in\Pi^\ast(S)\}\subseteq\mathbb{R}^{\Pi^\ast(V)}.
\end{equation}
Similarly, for every hyperedge $e_i$, we can write $U_{e_i}$ as the span of all matrix columns $\mathbf{1}_{e_i,\sigma}, \sigma\in\Pi^\ast(e_i)$. Therefore, the column space of $R(H)$ can be written as
$\operatorname{Col}(R(H))=\sum_{i=1}^m U_{e_i}$,\footnote{In Appendices A, B and C, the summation symbol $\Sigma$ over subspaces always refers to their linear span. } where $m$ is the total number of hyperedges. The rank of $R(H)$ is then given by:
\begin{equation}\label{rankR}
\operatorname{rank}_{\mathbb{R}}\bigl(R(H)\bigr)=\dim\left(\sum_{i=1}^m U_{e_i}\right),
\end{equation}
and we focus on deriving an inclusion-exclusion form for the desired quantity $\dim\left(\sum_{i=1}^m U_{e_i}\right)$.

Before further calculations, we first consider the physical interpretation of $U_{e_i}$. In fact, for every $f \in \mathbb{R}^{\Pi^\ast(V)}$ and every $ S\subseteq V$ with $|S|\ge2$,
\begin{align}
    f\in U_S \iff & 1)\ \forall P, Q \in \Pi^\ast(V), P|_S=Q|_S\rightarrow f(P)=f(Q) \\ 
    & 2)\ \forall P|_S={\pi_0}_S, f(P)=0,
\end{align}
where for any vertex set $T$, we denote ${\pi_0}_T$ as the trivial partition on it. Namely, $U_{S}$ is the collection of all the functions acting on $P\in \Pi^\ast(V)$ such that i) they depend only on the restriction of $P$ on $S$ and ii) they vanish for any $P$ that is trivial on $S$.
From the definition of $U_S$ and $\mathbf{1}_{S,\pi}$, any $f\in U_S$ naturally satisfies conditions $1)$ and $2)$ on the right-hand side. 

We first give a proper definition of the ``singleton extension": For vertex sets $A\subseteq B\subseteq V$ and $\sigma\in \Pi(A)$, we denote by $\widetilde{\sigma}^B$, the extension of $\sigma$ to $\Pi(B)$ that satisfies $\widetilde{\sigma}^B|_A=\sigma$ while including every vertex in $B/A$ as a singleton block\footnote{Here $B/A$ denotes the set of elements in $B$ but not in $A$. When $B=A$, $B/A=\emptyset$ and $\widetilde{\sigma}^B=\sigma$.}.
Returning to the proof, in the other direction, if $S\subsetneq V$, for every $\sigma _j \in \Pi(S)$, there always exists a $K_j=\widetilde{\sigma_j}^V \in \Pi^\ast(V)$ with $K_j|_S=\sigma_j$. Denoting $a_j=f(K_j)$, we can introduce the function 
\begin{equation}
    g=\sum_{\sigma_j \in \Pi(S)}a_j\mathbf{1}_{S,\sigma_j}.
\end{equation}
For any $P \in \Pi^\ast(V)$, there is always a $\sigma _k\in \Pi(S)$ with $P|_S=\sigma _k$. Because of $1)$, $g(P)=f(P)=a_k$. Meanwhile, by $2)$, the coefficient $a_0$ corresponding to $\sigma_0={\pi_0}_S$ vanishes, and we can conclude that 
\begin{equation}
    f=g=\sum_{\sigma_j \in \Pi^\ast (S)}a_j\mathbf{1}_{S,\sigma_j}\in U_S.
\end{equation}
If $S=V$, we directly have $f=\sum_{K_j \in \Pi^\ast (V)}f(K_j)\mathbf{1}_{S,K_j}\in U_V$, completing the proof of conditions 1) and 2).

We then introduce some crucial properties satisfied by $U_S$ that greatly simplify the proof of Theorem \ref{thm1}.
To begin with, for any $S\subseteq V$, 
\begin{equation}\label{USdim}
    \operatorname{dim}(U_S)=B_{|S|}-1.
\end{equation}
Namely, the dimension of every $U_S$ is given by the Bell number of $|S|$ minus one. As discussed in Definition 1, $B_{|S|}-1$ calculates the number of nontrivial partitions $\pi \in \Pi^\ast(S)$. When $|S|\leq 1$, there is no nontrivial $\pi \in \Pi^\ast(S)$, so $\operatorname{dim}(U_S)=0=B_{|S|}-1$. When $|S|\ge2$, every $\mathbf{1}_{S,\pi} \in U_S$ is nonzero, since for $\pi \in \Pi^\ast(S)$, one always has the singleton expansion $P=\widetilde{\pi}^V\in \Pi^\ast(V)$ with $P|_S=\pi$. By definition, $\mathbf{1}_{S,\pi}(P)=1$. Furthermore, for any $\pi, \pi' \in \Pi^\ast(S)$, the supports $\{P|P\in \Pi^\ast(V): P|_S=\pi\}$ and $\{P'|P'\in \Pi^\ast(V): P|_S=\pi'\}$ are disjoint, rendering the $\{\mathbf{1}_{S,\pi}\}_{\pi \in \Pi^\ast(S)}$ linearly independent. Thus, they form a valid basis for $U_S$, and (\ref{USdim}) is proved\footnote{It is easy to see that, when $S\subsetneq V$, similarly, $\mathbf{1}_{S,{\pi_0}_S}$ is also nonzero and has disjoint supports with $\{\mathbf{1}_{S,\pi}\}_{\pi \in \Pi^\ast(S)}$. Therefore, the complete $\{\mathbf{1}_{S,\pi}\}_{\pi \in \Pi(S)}$ is also linearly independent.}.

Based on (\ref{USdim}), the right-hand side of Theorem \ref{thm1} can be written as:
\begin{equation}\label{RHSasmeet}
    \sum_{\emptyset\neq F\subseteq[m]}(-1)^{|F|+1}(B_{k(F)}-1)=\sum_{\emptyset\neq F\subseteq[m]}(-1)^{|F|+1}\operatorname{dim}\big( \bigcap_{i\in F} U_{e_i}\big).
\end{equation}
For a cleaner presentation of the proof, we introduce the following lemma:
\begin{lemma}\label{MeetofU}
    $U_S \cap U_T=U_{S\cap T}$.
\end{lemma}
\begin{proof}
First we address the trivial cases where any of $U_S$, $U_T$ or $U_{S\cap T}$ may be $\{0\}$. When $U_S$ or $U_T$ is $\{0\}$, the equation is immediate. When $U_{S\cap T}=\{0\}$ but $U_S,U_T\neq\{0\}$, since $\{0\}\subseteq U_S \cap U_T$, we only need to prove the inclusion $U_S \cap U_T\subseteq U_{S\cap T}$. To begin with, any $f\in U_S\cap U_T$ can be expressed as:
\begin{equation}
    f=\sum_{\sigma_i \in \Pi^\ast (T)}a_i\mathbf{1}_{T,\sigma_i}.
\end{equation}
Under our premise, $|S\cap T|=0$ or $1$, and $T/S\neq\emptyset$\footnote{This is because $U_{S\cap T}=\{0\}\neq U_T$, so $|T|\ge 2$, while $|S\cap T|\le 1$.}.
Due to the limited overlap of $S$ and $T$, for any $\sigma_i\in \Pi^\ast (T)$ we can always construct some $P_i\in \Pi^\ast(V)$ such that $P_i|_S={\pi_0}_S$ and $P_i|_T=\sigma_i$\footnote{When $|S\cap T|=0$, the existence of $P_i$ is immediate. When $|S\cap T|=1$, stacking all vertices in $S/T$ into the same partition block as the single vertex in $S\cap T$ yields the desired $P_i$.}. By condition 2) on $S$, any such $f(P_i)=a_i=0$. Therefore, $f=0$, and we have $U_S \cap U_T\subseteq U_{S\cap T}$.

Consider then the case where none of $U_S$, $U_T$ and $U_{S\cap T}$ are $\{0\}$. Namely, $|S|,|T|,|S\cap T|\ge2$. For the inclusion $U_S \cap U_T\supseteq U_{S\cap T}$, any $g \in U_{S\cap T}$ has the decomposition 
\begin{equation}
    g=\sum_{\sigma_j \in \Pi^\ast (S\cap T)}a_j\mathbf{1}_{S\cap T,\sigma_j}.
\end{equation}
Note that by definition, $\mathbf{1}_{S\cap T,\sigma_j}(P)=1$ if and only if the $\pi=P|_{S}$ satisfies $\pi|_{S\cap T}=\sigma_j$. More explicitly, 
\begin{equation}
    \mathbf{1}_{S\cap T,\sigma_j}=\sum_{\substack{\pi\in\Pi^\ast(S)\\ \pi|_{S\cap T}=\sigma_j}}\mathbf{1}_{S,\pi}.
\end{equation}
Thus, $g\in U_S$. Similarly, exchanging $S$ and $T$ gives $g\in U_T$, resulting in $U_S \cap U_T\supseteq U_{S\cap T}$.

In the opposite direction, we use conditions $1)$ and $2)$ for $U_{S\cap T}$ to prove $f\in U_{S\cap T}$ for any $f\in U_S \cap U_T$. For condition $1)$, consider any $P,Q \in \Pi^\ast(V)$ such that $P|_{S\cap T}=Q|_{S\cap T}\neq {\pi_0}_{S\cap T}$\footnote{It suffices to discuss this case, since when $P|_{S\cap T}=Q|_{S\cap T}= {\pi_0}_{S\cap T}$, condition $1)$ follows directly from condition $2)$.}. Then one can always find an $R\in \Pi^\ast(V)$ satisfying $R|_{S}=P|_{S}$ and $R|_{T}=Q|_{T}$, since there is no contradiction between $P$ and $Q$ on $S\cap T$. Applying condition $1)$ on $U_S$ gives $f(R)=f(P)$. Similarly, from $U_T$ we have $f(R)=f(Q)$. As a result, $f(P)=f(R)=f(Q)$, and condition $1)$ for $f\in U_{S\cap T}$ is proved.

Regarding condition $2)$, for any $P$ satisfying $P|_{S\cap T}={\pi_0}_{S\cap T}$, if $P|_{S}={\pi_0}_{S}$ (or $P|_{T}={\pi_0}_{T}$), then by condition $2)$ on $U_S$ (or $U_T$), $f(P)=0$. On the other hand, if $P|_{S}\neq {\pi_0}_{S}$, and $P|_{T}\neq {\pi_0}_{T}$, we can define an $R'\in \Pi^\ast(V)$ by moving all vertices in $S/T$ into the partition block of $P$ that contains $S\cap T$\footnote{Since $P|_{S}\neq 0$ and $P|_{S\cap T}={\pi_0}_{S\cap T}$, we have $S/T\neq \emptyset$.}. This leads to $R'|_{T}=P|_{T}$, and $R'|_S={\pi_0}_{S}$. By condition $2)$ of $U_S$, $f(R')=0$, while condition $1)$ of $U_T$ implies $f(P)=f(R')$. Thus, $f(P)=0$ and we have proved the condition $2)$ for $f\in U_{S\cap T}$. 
Combining $1)$ and $2)$ gives the reverse inclusion
relation $U_S \cap U_T\subseteq U_{S\cap T}$. In conclusion, $U_S \cap U_T= U_{S\cap T}$\footnote{In the following discussion, we identify $U_S\cap U_T$ with $U_{S\cap T}$, and make no distinction between them.}.

\end{proof}

Therefore, from (\ref{rankR}), (\ref{RHSasmeet}) and Lemma \ref{MeetofU}, to prove Theorem \ref{thm1} it suffices to show
\begin{equation}\label{simplifiedthm1}
    \dim\Big(\sum_{i=1}^{m}U_{e_i}\Big)
=\sum_{\emptyset\neq F\subseteq [m]}(-1)^{|F|+1}\dim\bigl(U_{\bigcap_{i\in F} e_i}\bigr).
\end{equation}
During this proof, a stronger version of Lemma \ref{MeetofU} concerning the linear combinations of $U_{T_i}$ plays a significant role.
\begin{lemma}\label{MeetofSumU}
\begin{equation}\label{MeetSum}
    U_S\ \cap \big(\sum_{i=1}^{m} U_{T_i} \big)=\sum_{i=1}^{m} U_{S\cap T_i}.
\end{equation}
\end{lemma}
\begin{proof}
    First we need to clear up some subtleties. When $S=V$, (\ref{MeetSum}) becomes the trivial identity $\sum_{i=1}^{m} U_{T_i}=\sum_{i=1}^{m} U_{T_i}$, so we only consider $S\subsetneq V$. Moreover, when $|S|\le1$, $U_S=\{0\}$, and both sides of the equation are $\{0\}$. When some $U_{T_i}$ are $\{0\}$, they contribute nothing to the sums on either sides, so it is safe to discard them. Therefore, it is sufficient to focus on the case where some $U_{S\cap{T_i}}=\{0\}$ and some are not, which we will elaborate on later.

    For the inclusion  ``$\supseteq$” direction, by Lemma \ref{MeetofU}, we have
    \begin{equation}
        U_{S\cap T_i}=U_S\cap U_{T_i}\subseteq U_S\ \cap \big(\sum_{i=1}^{m} U_{T_i} \big).
    \end{equation}
    Taking the sum over $i$,
    \begin{equation}
    U_S\ \cap \big(\sum_{i=1}^{m} U_{T_i} \big)\supseteq\sum_{i=1}^{m} U_{S\cap T_i}.
    \end{equation}
    
    Then it remains to prove the ``$\subseteq$" direction. For any $f\in U_S\ \cap \big(\sum_{i=1}^{m} U_{T_i} \big)$, we need to show that $f\in \sum_{i=1}^{m} U_{S\cap T_i}$. 
    Recalling the definition of the singleton extension, since $S\subsetneq V$, for any $\pi\in \Pi(S)$, there is always a $P_\pi=\widetilde{\pi}^V\in \Pi^\ast(V)$ with $P_\pi|_S=\pi$. According to condition 1), since $f\in U_S$, specifying the value of $f(P_\pi)$ for every $\pi\in \Pi(S)$ is sufficient to define $f$\footnote{Although condition 2) on $U_S$ stipulates that for any $P$ with $P|_S={\pi_0}_S$, $f(P)=0$, we will use it later, and include $\pi={\pi_0}_S$ in our discussion for now.}. 
    At the same time, $f\in \sum_{i=1}^{m} U_{T_i}$, so it can be written as $f=\sum_{i=1}^{m} g_i$, with $g_i\in U_{T_i}$. Again, condition 1) implies that for any $P\in \Pi^\ast(V)$, $g_i(P)$ depends only on $P|_{T_i}$. 
    Denoting $A_i=S\cap T_i$ for each $T_i$, we first divide the edge indices $i$ into two sets:
    \begin{equation}
        I_{\ge2}=\{i:|A_i|=|S\cap T_i|\ge 2\},\ \text{and}\ I_{\le1}=\{i:|A_i|=|S\cap T_i|\le 1\}.
    \end{equation}
    For every $g_i$ with $i\in I_{\ge2}$ and every $P_\pi$, 
    \begin{equation}g_i(P_\pi)=g_i(\widetilde{(\pi|_{A_i})}^{V})
    \end{equation}
    This is because the restriction of $P_\pi$  and $\widetilde{(\pi|_{A_i})}^{V}$ on $T_i$ both equal $\pi|_{A_i}$ on vertices in $A_i$, and consist of singleton blocks for vertices in $T_i/A_i$. 
    We also define $P_{1,V}$ as the ``singleton partition" of $V$, where every vertex in $V$ forms a single block. Hence similarly, in the case where $|A_i|\le1$ and $U_{S\cap{T_i}}=\{0\}$,
     \begin{equation}g_i(P_\pi)=g_i(P_{1,V}),
    \end{equation}
    because $T_i$ contains at most one vertex from $S$, and every other vertex forms a singleton block. 
    Therefore, \begin{equation}\label{fPpi}
    f(P_\pi)=\sum_{i\in I_{\ge2}} g_i(\widetilde{(\pi|_{A_i})}^{V})+\sum_{i\in I_{\le1}} g_i(P_{1,V}).
    \end{equation}
    Note that the second term is constant and independent of $P_\pi$. Thus, the right-hand side of (\ref{fPpi}) relies only on $\pi|_{A_i}$. Namely, for any $P\in \Pi^\ast(V)$ with $\pi=P|_S$, the value $f(P)=f(P_\pi)$ depends solely on $P|_{A_i}$. 
    Also, let $\mathbf {1}$ denote the all-ones function in $\mathbb R^{\Pi^\ast(V)}$. For $|A_i|\ge2$, $\sum_{\tau\in\Pi(A_i)}\mathbf 1_{A_i,\tau}=\mathbf 1$
    \footnote{This is because for any $P\in \Pi^\ast(V)$, there exists a unique partition $\tau=P|_{A_i}$. 
    Thus, $\sum_{\tau\in\Pi(A_i)}\mathbf 1_{A_i,\tau}(P)$ is always $1$.}. 
    Similar to the proof of conditions 1) and 2), $f$ can be written in the form:
    \begin{equation}\label{ffortau}
         f = \sum_{i\in I_{\ge2}} \sum_{\tau\in \Pi(A_i)}g_i(\widetilde{\tau}^{V})\mathbf{1}_{A_i, \tau}+\sum_{i\in I_{\le1}} g_i(P_{1,V})\mathbf{1}.
    \end{equation}
    Imposing condition 2) for $f\in U_S$, we have
    $f(\widetilde{{\pi_0}_{S}}^{V})=0$. Substituting this into (\ref{ffortau}) gives 
    \begin{equation}\label{trivial0}
    \sum_{i\in I_{\ge2}} g_i(\widetilde{{\pi_0}_{A_i}}^{V})+\sum_{i\in I_{\le1}} g_i(P_{1,V})=0.  
    \end{equation}
    Therefore, (\ref{trivial0}) leads to:
    \begin{align*}
        f&=  \sum_{i\in I_{\ge2}} \big(\sum_{\tau\in \Pi(A_i)}g_i(\widetilde{\tau}^{V})\mathbf{1}_{A_i, \tau}-g_i(\widetilde{{\pi_0}_{A_i}}^{V})\mathbf{1}\big)\\
        &=\sum_{i\in I_{\ge2}} \sum_{\tau\in \Pi(A_i)}\big(g_i(\widetilde{\tau}^{V})-g_i(\widetilde{{\pi_0}_{A_i}}^{V})\big)\mathbf{1}_{A_i, \tau}\\
        &=\sum_{i\in I_{\ge2}} \sum_{\tau\in \Pi^\ast(A_i)}\big(g_i(\widetilde{\tau}^{V})-g_i(\widetilde{{\pi_0}_{A_i}}^{V})\big)\mathbf{1}_{A_i, \tau}.
    \end{align*} 
    The last equality is because for $\tau=\widetilde{{\pi_0}_{A_i}}$, the coefficient $\big(g_i(\widetilde{\tau}^{V})-(g_i(\widetilde{{\pi_0}_{A_i}}^{V})\big)$ vanishes. Thus, all $\mathbf{1}_{A_i, {\pi_0}_{T_i}}$ have been eliminated, and $f\in \sum_{i\in I_{\ge2}} U_{A_i}=\sum_{i=1}^m U_{S\cap T_i}$, since $\sum_{i\in I_{\le1}} U_{A_i}$ contributes nothing. This completes the proof of the ``$\subseteq$" direction.
\end{proof}

Finally, we move on to prove (\ref{simplifiedthm1})\footnote{To avoid subtleties in the following proof, we demand that (\ref{simplifiedthm1}) includes the case where some $U_{e_i}$ can be $\{0\}$.}: 
\begin{equation*}
\operatorname{dim}\Big(\sum_{i=1}^{m}U_{e_i}\Big)
=\sum_{\emptyset\neq F\subseteq [m]}(-1)^{|F|+1}\operatorname{dim}\bigl(U_{\bigcap_{i\in F} e_i}\bigr).
\end{equation*}
We proceed by induction on $m$. First, when $m=1$, both sides give $\operatorname{dim}(\sum_{i=1}^{m}U_{e_i})$, and (\ref{simplifiedthm1}) holds. Then, assuming (\ref{simplifiedthm1}) to hold for $m-1$, we denote 
\begin{equation}
    A:=\sum_{i=1}^{m-1}U_{e_i},
\quad
B:=U_{e_m}
\end{equation}
and check the validity of (\ref{simplifiedthm1}) for the case of $m$. For the left-hand side, 
\begin{equation}
    \operatorname{dim}\Big(\sum_{i=1}^{m}U_{e_i}\Big)=\operatorname{dim}(A+B).
\end{equation}
As finite-dimensional subspaces, $A$ and $B$ satisfy
\begin{equation}\label{dimensionplus}
    \operatorname{dim}(A+B)=\operatorname{dim}(A)+\operatorname{dim}(B)-\operatorname{dim}(A\cap B).
\end{equation}
According to Lemma \ref{MeetofSumU}, 
\begin{equation}
    A\cap B=\big(\sum_{i=1}^{m-1}U_{e_i}\big)\cap U_{e_m}=\sum_{i=1}^{m-1}U_{e_i\cap e_m}.
\end{equation}
Applying the assumed (\ref{simplifiedthm1}) for $m-1$ to both $A$ and $A\cap B$, we obtain
\begin{align}
    &\operatorname{dim}(A)=\operatorname{dim}\Big(\sum_{i=1}^{m-1}U_{e_i}\Big)
=\sum_{\emptyset\neq F\subseteq [m-1]}(-1)^{|F|+1}\operatorname{dim}\bigl(U_{\bigcap_{i\in F} e_i}\bigr),\\
&\operatorname{dim}(A\cap B)=\operatorname{dim}\Big(\sum_{i=1}^{m-1}U_{e_i\cap e_m}\Big)
=\sum_{\emptyset\neq F\subseteq [m-1]}(-1)^{|F|+1}\operatorname{dim}\bigl(U_{\bigcap_{i\in F} (e_i\cap e_m)}\bigr).
\end{align}
In the latter expression, 
\begin{equation}
    \bigcap_{i\in F} (e_i\cap e_m)=(\bigcap_{i\in F} e_i)\cap e_m.
\end{equation}
Therefore, 
\begin{equation}
    \operatorname{dim}(A\cap B)=\sum_{\substack{|F|\ge2,m\in F\\ F\subseteq [m]}}(-1)^{|F|}\operatorname{dim}\Big(U_{\bigcap_{i\in F} e_i}\Big).
\end{equation}
Substituting these results into (\ref{dimensionplus}) yields:
\begin{align}
    \operatorname{dim}\Big(\sum_{i=1}^{m}U_{e_i}\Big)=&\sum_{\emptyset\neq F\subseteq [m-1]}(-1)^{|F|+1}\operatorname{dim}\bigl(U_{\bigcap_{i\in F} e_i}\bigr)+\operatorname{dim}(U_{e_m})\\
    &+\sum_{\substack{|F|\ge2,m\in F\\ F\subseteq [m]}}(-1)^{|F|+1}\operatorname{dim}\Big(U_{\bigcap_{i\in F} e_i}\Big).
\end{align}
We can see that the first term covers all subsets $F\subseteq[m]$ that do not contain $m$. The second term corresponds solely to the $F=\{m\}$ (for which $(-1)^{|F|+1}=1$), while the third term includes all subsets $F\subseteq[m]$ with $m\in F$ and $|F|\ge2$. Hence, the right-hand side is precisely $\sum_{\emptyset\neq F\subseteq [m]}(-1)^{|F|+1}\operatorname{dim}\bigl(U_{\bigcap_{i\in F} e_i}\bigr)$, and (\ref{simplifiedthm1})  holds for $m$. As a result, (\ref{simplifiedthm1}) is valid for all $m$, and the inclusion-exclusion formula for the rank of the partition-reduction matrix,
\begin{equation} \operatorname{rank}_{\mathbb{R}}\bigl(R(H)\bigr)=\sum_{\emptyset\neq F\subseteq[m]}(-1)^{|F|+1}(B_{k(F)}-1) \end{equation} is proved.

\section{Uniqueness of the equality set for a Sperner class}\label{appB}
This appendix presents the proof of  Theorem \ref{thm2}, which reveals the uniqueness of equalities that hold for a class of Sperner states. Namely, no different hypergraphs can share exactly the same set of linear equalities in MEMS$_n$.

In Section \ref{sec2}, we  reveal the capture of the linear dependence of the macroscopic coordinates by the equalities. Therefore, we can write:
\begin{equation}
    \operatorname{ker}({R(H})^{\mathsf T})=\mathrm{Col}(R(H))^{\perp},
\end{equation}
with the left being the kernel of the transposed partition-reduction matrix $R(H)^{\mathsf T}$ (or equivalently, the left null space) and the right being the orthogonal complement of the column space $\mathrm{Col}(R(H))$. This equation is a standard result from finite-dimensional linear algebra. Therefore, if for two hypergraphs $H_1$ and $H_2$, $\operatorname{ker}({R(H_1})^{\mathsf T})=\operatorname{ker}({R(H_2})^{\mathsf T})$, taking orthogonal complements gives
\begin{equation}\label{coleq}
    \mathrm{Col}(R(H_1))=\mathrm{Col}(R(H_2)).
\end{equation}
Meanwhile, according to Appendix \ref{appA}, by viewing the column space of partition-reduction matrices as a function space,  (\ref{coleq}) can now be written as:
\begin{equation}\label{Ueq}
    \sum_{e_i\in E_1}U_{e_i}=\sum_{e_j\in E_2}U_{e_j}.
\end{equation}
Based on this, we now propose a method to read off the set of hyperedges directly from the knowledge of $\sum_{e_i\in E}U_{e_i}$. As a result, for any given equality set, the hypergraph is unique. We begin with the following lemma:
\begin{lemma}\label{propersub}
    For any $S\subseteq V$ with $|S|\ge2$, and any $T_1,...,T_m\subsetneq S$, 
    \begin{equation}\label{propersubsetU}
        \sum_{i}^m U_{T_i}\subsetneq U_S.
    \end{equation}
\end{lemma}
\begin{proof}
    First we clear away some subtleties. When $|T_i|\le1$, $U_{T_i}=\{0\}$, and such terms do not affect the sum. Thus, it suffices to assume $|T_i|\ge 2$ for every $T_i$. Also, when $|S|=2$, all proper subsets $T_1,...,T_m\subsetneq S$ are either empty or singletons. In this case the left-hand-side of (\ref{propersubsetU}) is $\{0\}$ and (\ref{propersubsetU}) holds trivially. Therefore, it is enough to consider only $|S|\ge3$.

    To proceed with the proof, we introduce a linear isomorphism that transfers the inclusion (\ref{propersubsetU}) to a simpler family of spaces, where it can be established more easily.
    First recall the singleton extension $\widetilde{(\cdot)}^B: \Pi(A)\rightarrow\Pi(B)$ defined in Appendix \ref{appA}, which transforms a partition on $A\subseteq B$ into one on $B$ by adding the vertices in $B/A$ as singleton blocks. Note that $\widetilde{(\cdot)}^B$ is an injection.
    
    We then define a linear map $\Phi_S:U_S\rightarrow\mathbb{R}^{\Pi^\ast(S)}$: For any $f\in U_S$ and any $\pi\in \Pi^\ast(S)$,
    \begin{equation}
        \Phi_S(f)(\pi)=f(\widetilde{\pi}^V).
    \end{equation}
    In other words, $\Phi_S$ identifies a function on partitions on $V$ that depends only on their restriction to $S$ with a function on partitions on $S$, discarding the irrelevant information on $V/S$.
    To begin with, for any $\mathbf{1}_{S,\sigma}\in U_S$ and any $\pi\in \Pi^\ast(S)$,
    \begin{equation}
        \Phi_S(\mathbf{1}_{S,\sigma})(\pi)=\mathbf{1}_{S,\sigma}(\widetilde{\pi}^V)=\delta_{\sigma, \pi}=e_\pi(\sigma),
    \end{equation}
    where $\{e_\pi\}$ is the standard basis of $\mathbb{R}^{\Pi^\ast(S)}$. Therefore, $\Phi_S$ sends the basis $\{\mathbf{1}_{S,\sigma}\}$ of $U_S$ to the standard basis $\{e_\pi\}$ of $\mathbb{R}^{\Pi^\ast(S)}$, and is consequently a linear isomorphism.

    Moreover, we are also interested in the action of $\Phi_S$ on $U_{T_i}$. Note that when $T_i\subseteq S$, Lemma \ref{MeetofU} implies $U_{T_i}\subseteq U_S$, so that $\Phi_S$ can also act on $U_{T_i}$. Given any $T_i\subseteq S$ with $|T_i|\ge2$, any $\mathbf{1}_{T_i,\tau}\in U_{T_i}$ and any $\pi\in \Pi^\ast(S)$,
    \begin{equation}
        \Phi_S(\mathbf{1}_{T_i,\tau})(\pi)=\mathbf{1}_{T_i,\tau}(\widetilde{\pi}^V)=\delta_{\tau, \pi|_{T_i}}.
    \end{equation}
    We denote by $\overline{\mathbf{1}_{T_i,\tau}}\in \mathbb{R}^{\Pi^\ast(S)}$ the function defined by $\overline{\mathbf{1}_{T_i,\tau}}(\pi)=\delta_{\tau, \pi|_{T_i}}$ for any $\pi\in \Pi^\ast(S)$. Thus, $\Phi_S(\mathbf{1}_{T_i,\tau})=\overline{\mathbf{1}_{T_i,\tau}}$. Moreover, we introduce\footnote{Note the similarity between $\overline{U_{T_i}}, \overline{\mathbf{1}_{T_i,\tau}}$, and ${U_{T_i}}, {\mathbf{1}_{T_i,\tau}}$. The only difference is that the former definitions use $S$ in place of $V$.} 
    \begin{equation}
        \overline{U_{T_i}}=\operatorname{span}\{\overline{\mathbf{1}_{T_i,\tau}}:\tau \in \Pi^\ast (T_i)\}\subseteq \mathbb{R}^{\Pi^\ast(S)}.
    \end{equation}
    In summary, $\Phi_S$ maps $U_S$ to $\mathbb{R}^{\Pi^\ast(S)}$, and $U_{T_i}$ to $\overline{U_{T_i}}$. Therefore, by the properties of linear isomorphisms, to prove (\ref{propersubsetU}), it is sufficient to show:
    \begin{equation}\label{sufUTS}
        \sum_i^m\overline{U_{T_i}}\subsetneq\mathbb{R}^{\Pi^\ast(S)}.
    \end{equation}

    For (\ref{sufUTS}) to hold, it is enough to find a function $w\in\mathbb{R}^{\Pi^\ast(S)}$ that does not belong to $\sum_i^m\overline{U_{T_i}}$. 
    Define 
    \begin{equation}
        w(\pi)=(-1)^{s-|\pi|}(|\pi|-1)!\in \mathbb{R}^{\Pi^\ast(S)},\ \pi\in \Pi^\ast(S),
    \end{equation}
    where we denote $s=|S|$, and $|\pi|$ is the number of partition blocks of $\pi$. First, we observe that $w\neq0$, because when $|S|\ge3$, there always exists a nontrivial bipartition $\pi\in \Pi^\ast(S)$ such that $w=(-1)^{s-2}\neq0$. 
    Consider the standard inner product on $\mathbb{R}^{\Pi^\ast(S)}$:
    \begin{equation}
        \langle a,b\rangle=\sum_{\pi\in\Pi^\ast(S)}a(\pi)b(\pi).
    \end{equation}
    We claim that for any $T_i\subsetneq S$ with $|T_i|\ge2$ and any $\tau\in\Pi^\ast(T_i)$,
    \begin{equation}\label{wperp}
        \langle w,\overline{\mathbf{1}_{T_i,\tau}}\rangle=\sum_{\pi\in\Pi^\ast(S)}w(\pi)\overline{\mathbf{1}_{T_i,\tau}}(\pi)=\sum_{\substack{\pi\in\Pi^\ast(S),\\\pi|_{T_i}=\tau}}(-1)^{s-|\pi|}(|\pi|-1)!=0,
    \end{equation}
    where the second equation follows from the definition of $\overline{\mathbf{1}_{T_i,\tau}}$.

    To prove (\ref{wperp}), the summation should first be written in a more explicit form. We denote $t=|T_i|$, $r=|S/T_i|=s-t\ge1$, and $b=|\tau|$ as the number of partition blocks in the partition $\tau$. Note that any $\pi\in\Pi^\ast(S)$ with $\pi|_{T_i}=\tau$ can be uniquely constructed in the following way:

    First, divide the vertex set $S/T_i$ into two parts: $U$ and $(S/T_i)/U$. The vertices in $U$ are partitioned into $j$ blocks, which remain independent blocks of $\pi$ disjoint from those in $\tau$. On the other hand, each vertex in $(S/T_i)/U$ is attached to one of the $b$ blocks in $\tau$. Thus, we obtain all partition blocks of $\pi$ that contain vertices both inside and outside $T_i$. With $\tau$ fixed, this method allows us to find every $\pi$ satisfying $\pi|_{T_i}=\tau$. Denoting $u=|U|$, the number of such $\pi$ with $|\pi|=b+j$ is:
    \begin{equation}\label{countpi}
        \binom{r}{u}\left\{ {u \atop j} \right\}b^{r-u},\ 0\le j\le u\le r,
    \end{equation}
    where $\left\{ {u \atop j} \right\}$ is the Stirling number of the second kind.
    Substituting (\ref{countpi}) into (\ref{wperp}) gives:\footnote{As long as $\tau$ is nontrivial, the $\pi$ thus obtained must be nontrivial and belong in $\Pi^\ast(S)$.}
    \begin{align}
        \langle w,\overline{\mathbf{1}_{T_i,\tau}}\rangle&=\sum_{u=0}^r\sum_{j=0}^u\binom{r}{u}\left\{ {u \atop j} \right\}b^{r-u}(-1)^{s-(b+j)}(b+j-1)!\\
        &=\sum_{u=0}^r\binom{r}{u}b^{r-u}(-1)^{s-b}\sum_{j=0}^u\left\{ {u \atop j} \right\}(-1)^{j}(b+j-1)!,
    \end{align}
    where we use $(-1)^{-j}=(-1)^j$ in the second line. Also, noting that
    \begin{equation}
        (b+j-1)!=(b+j-1)(b+j-2)...b(b-1)!,
    \end{equation}
    the inner sum can be written as:
    \begin{align}
        &\sum_{j=0}^u\left\{ {u \atop j} \right\}(-1)^{j}(b+j-1)!=(b-1)!\sum_{j=0}^u\left\{ {u \atop j} \right\}(-1)^{j}(b+j-1)(b+j-2)...b\\
        &=(b-1)!\sum_{j=0}^u\left\{ {u \atop j} \right\}(-b)(-b-1)...(-b-j+1)=(b-1)!\sum_{j=0}^u\left\{ {u \atop j} \right\}(-b)_j,
    \end{align}
    where $(x)_j=x(x-1)...(x-j+1)$ is the fallen factorial. By the standard identity that expresses ordinary powers in the basis of falling factorials, $x^u=\sum_{j=0}^u\left\{ {u \atop j} \right\}(x)_j$. Therefore, 
    \begin{align}
        &\quad \quad \quad \quad \langle w,\overline{\mathbf{1}_{T_i,\tau}}\rangle=\sum_{u=0}^r\binom{r}{u}b^{r-u}(-1)^{s-b}(b-1)!(-b)^u\\
        &=(-1)^{s-b}(b-1)!\sum_{u=0}^r\binom{r}{u}b^{r-u}(-b)^u=(-1)^{s-b}(b-1)!(b+(-b))^r=0
    \end{align}
    because $r=|S/T_i|\ge1$. As a result, (\ref{wperp}) is proved. Since $T_i$ and $\tau$ are arbitrary, $w$ is orthogonal to every $\overline{U_{T_i}}$, hence
    \begin{equation}
        w\perp\sum_{i=1}^m \overline{U_{T_i}}.
    \end{equation}
    As $w\neq 0$, it follows from properties of the inner product that $\sum_{i=1}^m\overline{U_{T_i}}$ cannot equal $\mathbb{R}^{\Pi^\ast(S)}$, so (\ref{sufUTS}) holds. 
    Finally, applying $\Phi_S^{-1}$ to (\ref{sufUTS}) yields:
    \begin{equation}
        \sum_{i}^m U_{T_i}\subsetneq U_S.
    \end{equation}
\end{proof}

We then show the method to uniquely obtain $E$ directly from $\mathrm{Col}(R(H))$. Since the specific hyperedges are unknown, we denote $W(H)=\sum_{e_i\in E}U_{e_i}=\mathrm{Col}(R(H))$.
First of all, we introduce a collection of vertex sets:
\begin{equation}
    D(H)=\{S\subseteq V: |S|\ge2 \ \text{and}\ U_S\subseteq W(H)\}.
\end{equation}
We will argue that this is exactly the vertex set collection:
\begin{equation}
    D'(H)=\{S\subseteq V: |S|\ge2 \ \text{and}\ S\subseteq e_i \ \text{for some}\ e_i\in E\}.
\end{equation}
To begin with, $D(H)\supseteq D'(H)$ is immediate. By Lemma \ref{MeetofU}, if $S\subseteq e_i$ for some $e_i\in E$, then $U_S=U_{S\cap e_i}\subseteq U_{e_i}\subseteq W(H)$. Considering the opposite direction, for any $S\in D(H)$, $U_S\subseteq W(H)$, which leads to 
\begin{equation}
    U_S\cap W(H)=U_S\cap (\sum_{e_i\in E}U_{e_i})=U_S.
\end{equation}
According to Lemma \ref{MeetofSumU}, $U_S\cap (\sum_{e_i\in E}U_{e_i})=\sum_{e_i\in E}U_{S\cap e_i}$. Therefore,
\begin{equation}\label{Uscap}
    U_S=\sum_{e_i\in E}U_{S\cap e_i}.
\end{equation}
Notice that if $S$ is not contained in any $e_i$, every $S\cap e_i\subsetneq S$. By Lemma \ref{propersub}, this implies 
\begin{equation}
    \sum_{e_i\in E}U_{S\cap e_i}\subsetneq U_S,
\end{equation}
which contradicts with (\ref{Uscap}). Consequently, $S$ must be contained in some $e_i$, and we also have $D(H)\subseteq D'(H)$. Combining both directions, we conclude $D(H)= D'(H)$.
Now, for any given $W(H)$, the vertex set collection $D(H)$ can be naturally constructed. In fact, the maximal elements in $D(H)$ under inclusion are exactly the hyperedges we are searching for. Indeed, for every $e_i\in E$, $U_{e_i}\subseteq W(H)$, so $e_i\in D(H)$. Meanwhile, since $D(H)=D'(H)$, the maximal elements of $D(H)$ are also subsets some $e_i\in E$. Thus, these elements are precisely the hyperedges. Combining them into an antichain yields our desired set $E$.

In summary, for any given $\operatorname{ker}({R(H})^{\mathsf T})=\mathrm{Col}(R(H))^\perp$, its orthogonal complement $\mathrm{Col}(R(H))=\sum_{e_i\in E}U_{e_i}=W(H)$ is unique, and using the method above, the set of hyperedges $E$ can be uniquely recovered from $W(H)$. Therefore, if $\operatorname{ker}({R(H_1})^{\mathsf T})=\operatorname{ker}({R(H_2})^{\mathsf T})$, the hyperedges thus obtained must satisfy $E_1=E_2$, and the corresponding hypergraphs are identical. In other words, the same set of linear equalities cannot be shared by two different hypergraphs. 
\section{The union and intersection of
geometric objects and hypergraphs}\label{appC}
In this appendix we prove Theorem \ref{thm3} and Theorem \ref{thm4}, which give a direct description of the geometric objects corresponding to the union and the intersection of two hypergraphs. Due to the equivalence between these geometric objects and linear equalities, we present our proof in the language of the latter.

\subsection{Proof of Theorem \ref{thm3}}
First of all, Theorem \ref{thm3} equates the number of equalities of the union hypergraph $H_{a\lor b}$ to the set of equalities shared by the two separate hypergraphs $H_a,H_b$.
To prove this, we again reformulate the discussion in terms of column spaces. Following the setup in Appendix \ref{appB}, Theorem \ref{thm3} can be rephrased as:
\begin{equation}\label{thm3reph}
\mathrm{Col}(R(H_{a\vee b}))^{\perp}=\mathrm{Col}(R(H_{a}))^{\perp}\cap \mathrm{Col}(R(H_{b}))^{\perp}.
\end{equation}
For the left-hand side, we should first note that the column space
\begin{equation}\label{Colunion}
    \mathrm{Col}(R(H_{a\vee b}))=\mathrm{Col}(R(H_{ab})),
\end{equation}
where $H_{ab}$ is the simple union of $H_a$ and $H_b$ without removing the redundant edges\footnote{Although such $H_{ab}$ is not a proper hypergraph, the definitions and properties of the  corresponding $U_{e_j}$ with ${e_j}\in E_{a}\cup E_{b}$ can be extended to this case.}. 
Thus, 
\begin{equation}
\mathrm{Col}(R(H_{a\vee b}))=\mathrm{Col}(R(H_{ab}))=\sum_{e\in E_{a}\cup E_{b}}U_e=(\sum_{e\in E_{a}}U_e)+(\sum_{e\in E_{b}}U_e).
\end{equation}
Taking the orthogonal complement, we obtain
\begin{equation}
\mathrm{Col}(R(H_{a\vee b}))^\perp=\big ((\sum_{e\in E_{a}}U_e)+(\sum_{e\in E_{b}}U_e)\big)^\perp=(\sum_{e\in E_{a}}U_e)^{\perp}\cap (\sum_{e\in E_{b}}U_e)^{\perp}=\mathrm{Col}(R(H_{a}))^{\perp}\cap \mathrm{Col}(R(H_{b}))^{\perp},
\end{equation}
where the second equation follows directly from the standard property of orthogonal complements. This completes the proof of Theorem \ref{thm3}.

\subsection{Proof of Theorem \ref{thm4}}
We now turn to the equalities corresponding to the intersection of two hypergraphs. Theorem \ref{thm4} states that these equalities are precisely the ones appearing in either of the Sperner classes. 
Using the same notation as the previous proof, the theorem can be expressed as:
\begin{equation}
    \mathrm{Col}(R(H_{a\land b}))^\perp=\mathrm{Col}(R(H_{a}))^{\perp}+ \mathrm{Col}(R(H_{b}))^{\perp}.
\end{equation}
From the standard relation $Y^{\perp}+Z^{\perp}=(Y\cap Z)^{\perp}$ for subspaces $Y$, $Z$, we have       
\begin{equation}
    \mathrm{Col}(R(H_{a}))^{\perp}+ \mathrm{Col}(R(H_{b}))^{\perp}=(\mathrm{Col}(R(H_{a}))\cap \mathrm{Col}(R(H_{b})))^{\perp}.
\end{equation}
Hence, to prove Theorem \ref{thm4}, it suffices to show
\begin{equation}\label{thm4col}
    \mathrm{Col}(R(H_{a\land b}))=\mathrm{Col}(R(H_{a}))\cap \mathrm{Col}(R(H_{b})),
\end{equation}
which requires us to first introduce a more general version of Lemmas \ref{MeetofU} and \ref{MeetofSumU}:
\begin{lemma}\label{MeetofDoublesum}
    \begin{equation}\label{Meet2sum}
        (\sum_{i=1}^mU_{S_i})\cap(\sum_{j=1}^lU_{T_j})=\sum_{i=1}^m\sum_{j=1}^lU_{S_i\cap T_j}.
    \end{equation}
\end{lemma}
\begin{proof}
    We prove this lemma by induction on $m$. First, note that for $m=1$, (\ref{Meet2sum}) reduces to $U_{S_i}\cap(\sum_{j=1}^lU_{T_j})=\sum_{j=1}^lU_{S_i\cap T_j}$, which is precisely Lemma \ref{MeetofSumU}.
    Now, assume that (\ref{Meet2sum}) is valid for $m-1$, and we will prove that it also holds for $m$.
    Denote
    \begin{equation}
        V=\sum_{i=1}^{m-1}U_{S_i},\ \ \ \ W=U_{S_m},\ \ \text{and}\ \ X=\sum_{j=1}^{l}U_{T_j}.
    \end{equation}
    Then the left-hand side of (\ref{Meet2sum}) can be written as $(V+W)\cap X$, while the right-hand side is:
    \begin{equation}
        \sum_{i=1}^m\sum_{j=1}^lU_{S_i\cap T_j}=\sum_{i=1}^{m-1}\sum_{j=1}^lU_{S_i\cap T_j}+\sum_{j=1}^lU_{S_m\cap T_j}=(V\cap X)+(W\cap X).
    \end{equation}
    For the last equation to hold, we apply respectively the induction hypothesis (\ref{Meet2sum}) for $m-1$, and Lemma  \ref{MeetofSumU} to the first and second terms. Thus, to complete the proof for Lemma  \ref{MeetofDoublesum}, it remains to show
    \begin{equation}\label{VWX}
        (V+W)\cap X=(V\cap X)+(W\cap X).
    \end{equation}
    The ``$\supseteq$" direction is immediate, since both $V\cap X$ and $W\cap X$ are contained in $(V+W)\cap X$. Meanwhile, for the opposite ``$\subseteq$" direction, we choose any $f\in(V+W)\cap X$, and try to prove $f\in(V\cap X)+(W\cap X)$. First, since $f\in V+W$, it can be decomposed as
    \begin{equation}
        f=v+w, \ v\in V,\ w\in W.
    \end{equation}
    Because $f\in X$, it follows that $w=f-v\in X+V$. As a result, 
    \begin{equation}
        w\in W\cap(V+X)=U_{S_m}\cap(\sum_{i=1}^{m-1}U_{S_i}+\sum_{j=1}^{l}U_{T_j}).
    \end{equation}
    Applying Lemma \ref{MeetofSumU} twice, the right-hand-side can be written as
     \begin{align}
         U_{S_m}\cap(\sum_{i=1}^{m-1}U_{S_i}+\sum_{j=1}^{l}U_{T_j})&=\sum_{i=1}^{m-1}U_{S_m}\cap U_{S_i}+\sum_{j=1}^{l}U_{S_m}\cap U_{T_j}\\ &=U_{S_m}\cap(\sum_{i=1}^{m-1} U_{S_i})+U_{S_m}\cap(\sum_{j=1}^{l} U_{T_j})\\
         &=W\cap V+W\cap X.
     \end{align}
    Therefore, we can write $w=w_V+w_X$, with $w_V\in W\cap V$ and $w_X\in W\cap X$ respectively. Meanwhile, $f=v+w=v+w_V+w_X$. Since both $f$ and $w_X$ belong to $X$, we must also have $v+w_V\in X$. Consequently, $v+w_V\in (V\cap X)$ and  $w_X\in (W\cap X)$, which implies $f\in (V\cap X)+(W\cap X)$.
    Combining the two inclusion directions, we have verified (\ref{VWX}), and therefore Lemma \ref{MeetofDoublesum} is proved\footnote{Note that here in the proof we use only Lemma \ref{MeetofSumU}, which is valid when any $U_S$, $U_{T_i}$ and $U_{S\cap T_i}$ are $\{0\}$. Therefore, if any $U_{S_i}$, $U_{T_j}$ or $U_{S_i\cap T_j}$ in Lemma \ref{MeetofDoublesum} are $\{0\}$, our argument still holds.}.
\end{proof}

Substituting $S_i=e_i\in E_a$ and $T_j=e_j\in E_b$ into Lemma \ref{MeetofDoublesum}, we obtain:
\begin{equation}\label{thm4U}
    (\sum_{e_i\in E_a}U_{e_i})\cap(\sum_{e_j\in E_b}U_{e_j})=\sum_{e_i\in E_a}\sum_{e_j\in E_b}U_{e_i\cap e_j}.
\end{equation}
The left-hand-side is exactly $\mathrm{Col}(R(H_{a}))\cap \mathrm{Col}(R(H_{b}))$, while the right-hand-side corresponds to the hypergraph whose hyperedges are the nontrivial
pairwise intersections of those in $E_a$ and $E_b$\footnote{When $|e_i\cap e_j|\le1$, $U_{e_i\cap e_j}=\{0\}$, so such intersections can be ignored.  }. This directly gives $\mathrm{Col}(R(H_{a\land b}))$. Thus, (\ref{thm4U}) confirms (\ref{thm4col}), and Theorem \ref{thm4} is proved.

\section{Examples of the equalities and hypergraphs for specific Sperner classes}\label{appD}
In this appendix, we introduce two classes of Sperner states and present their hypergraphs as well as the universal procedures to derive their corresponding equalities. We begin with the two hypergraphs: the $n=4$ 3-uniform-complete hypergraph and the $n=4$ 2-uniform-complete hypergraph, as shown in FIG \ref{example}.
\begin{figure} [h]
    \centering 
    \includegraphics[width=0.5\textwidth]{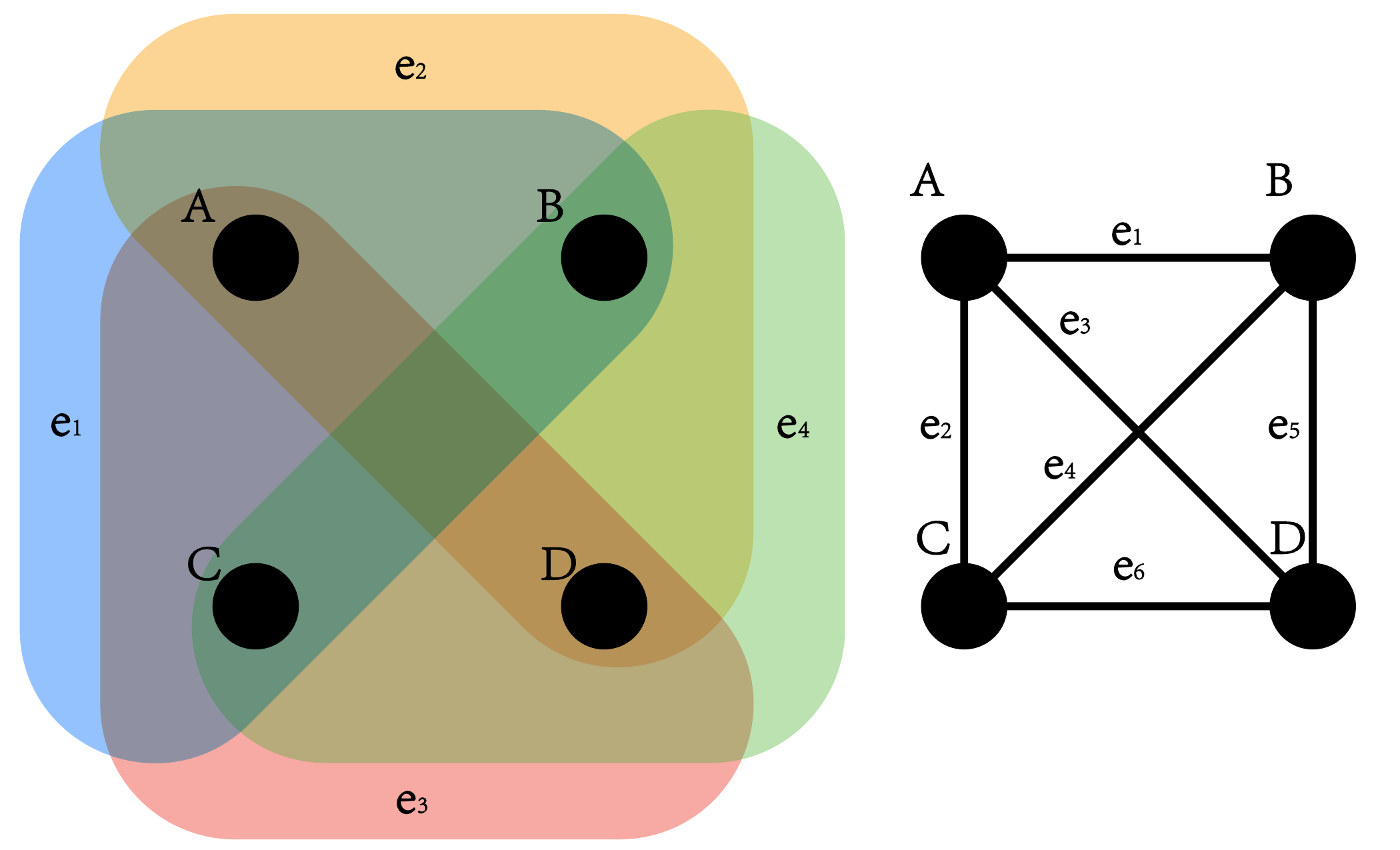} 
   \caption{Left: $n=4$ 3-uniform-complete hypergraph, each 3-hyperedge is represented as the colorful shaded area covering three black dots (vertices). Right: $n=4$ 2-uniform-complete hypergraph, each 2-hyperedge (edge) is represented as a line connecting two black dots (vertices).}
    \label{example} 
\end{figure}

The Sperner states associated with the $n=4$ 3-uniform-complete hypergraph, up to LU transformations, can be written as $\rho=\rho_{A_1B_2C_3}\otimes\rho_{B_1C_2D_3}\otimes\rho_{C_1D_2A_3}\otimes\rho_{D_1A_2B_3}$. Similarly, for the $n=4$ 2-uniform-complete hypergraph, the Sperner states are LU-equivalent to $\rho=\rho_{A_1B_3}\otimes\rho_{A_2C_2}\otimes\rho_{D_1A_3}\otimes\rho_{B_1C_3}\otimes\rho_{B_2D_2}\otimes\rho_{C_1D_3}$.
By Theorem \ref{thm1}, the codimension of their corresponding geometrical objects in MEMS$_n$ are $4$ and $8$ respectively. The partition-reduction matrix associated with the $n=4$ 3-uniform-complete hypergraph is shown in Table \ref{example43}.
\begin{table*}[ht]
\centering
\begin{tabular}{|c|c|c|c|c|c|c|c|c|c|c|c|c|c|c|c|c|}
\hline
$E^{(2)}_{(ABC|D)}$&$0$&$0$&$0$&$0$&$1$&$0$&$0$&$0$&$1$&$0$&$0$&$0$&$1$&$0$&$0$&$0$\\
\hline
$E^{(2)}_{(ABD|C)}$&$1$&$0$&$0$&$0$&$0$&$0$&$0$&$0$&$0$&$1$&$0$&$0$&$0$&$1$&$0$&$0$\\
\hline
$E^{(2)}_{(AB|CD)}$&$1$&$0$&$0$&$0$&$1$&$0$&$0$&$0$&$0$&$0$&$1$&$0$&$0$&$0$&$1$&$0$\\
\hline
$E^{(3)}_{(AB|C|D)}$&$1$&$0$&$0$&$0$&$1$&$0$&$0$&$0$&$0$&$0$&$0$&$1$&$0$&$0$&$0$&$1$\\
\hline
$E^{(2)}_{(B|ACD)}$&$0$&$1$&$0$&$0$&$0$&$1$&$0$&$0$&$0$&$0$&$0$&$0$&$0$&$0$&$1$&$0$\\
\hline
$E^{(2)}_{(AC|BD)}$&$0$&$1$&$0$&$0$&$0$&$0$&$1$&$0$&$1$&$0$&$0$&$0$&$0$&$1$&$0$&$0$\\
\hline
$E^{(3)}_{(AC|B|D)}$&$0$&$1$&$0$&$0$&$0$&$0$&$0$&$1$&$1$&$0$&$0$&$0$&$0$&$0$&$0$&$1$\\
\hline
$E^{(2)}_{(AD|BC)}$&$0$&$0$&$1$&$0$&$0$&$1$&$0$&$0$&$0$&$1$&$0$&$0$&$1$&$0$&$0$&$0$\\
\hline
$E^{(2)}_{(A|BCD)}$&$0$&$0$&$1$&$0$&$0$&$0$&$1$&$0$&$0$&$0$&$1$&$0$&$0$&$0$&$0$&$0$\\
\hline
$E^{(3)}_{(BC|A|D)}$&$0$&$0$&$1$&$0$&$0$&$0$&$0$&$1$&$0$&$0$&$0$&$1$&$1$&$0$&$0$&$0$\\
\hline
$E^{(3)}_{(AD|B|C)}$&$0$&$0$&$0$&$1$&$0$&$1$&$0$&$0$&$0$&$1$&$0$&$0$&$0$&$0$&$0$&$1$\\
\hline
$E^{(3)}_{(BD|A|C)}$&$0$&$0$&$0$&$1$&$0$&$0$&$1$&$0$&$0$&$0$&$0$&$1$&$0$&$1$&$0$&$0$\\
\hline
$E^{(3)}_{(CD|A|B)}$&$0$&$0$&$0$&$1$&$0$&$0$&$0$&$1$&$0$&$0$&$1$&$0$&$0$&$0$&$1$&$0$\\
\hline
$E^{(4)}_{(A|B|C|D)}$&$0$&$0$&$0$&$1$&$0$&$0$&$0$&$1$&$0$&$0$&$0$&$1$&$0$&$0$&$0$&$1$\\
\hline
\end{tabular}
\caption{The partition-reduction matrix associated with the $n=4$ 3-uniform-complete hypergraph.}
\label{example43}
\end{table*}

First, the macroscopic quantities, displayed as the row labels, are the measures evaluated on the entire 4-party state, corresponding to all bipartitons, tripartitions, and the single 4-partition of the vertices $A$, $B$, $C$ and $D$. Together, the values of these measures form $\textbf{V}^\mathrm{macro}$. Meanwhile, the columns are indexed by the microscopic quantities. For each 3-hyperedge, the measures for its three bipartitions and one tripartition are included among these quantities. Their values constitute $\textbf{V}^\mathrm{micro}$.
Specifically in the matrix, the columns correspond respectively to: ${e_1^{(2)}}_{(C|AB)}$, ${e_1^{(2)}}_{(B|AC)}$, ${e_1^{(2)}}_{(A|BC)}$, ${e_1^{(3)}}_{(A|B|C)}$, ${e_2^{(2)}}_{(D|AB)}$, ${e_2^{(2)}}_{(B|AD)}$, ${e_2^{(2)}}_{(A|BD)}$, ${e_2^{(3)}}_{(A|B|D)}$, ${e_3^{(2)}}_{(D|AC)}$, ${e_3^{(2)}}_{(C|AD)}$, ${e_3^{(2)}}_{(A|CD)}$, ${e_3^{(3)}}_{(A|C|D)}$, ${e_4^{(2)}}_{(D|BC)}$, ${e_4^{(2)}}_{(C|BD)}$, ${e_4^{(2)}}_{(B|CD)}$, and ${e_4^{(3)}}_{(B|C|D)}$. Note that in this appendix, we denote the microscopic quantities $E_{\pi_e}^{(|\pi_e|)}$ by lowercase $e$, with the subscript number labeling the hyperedge. The partitions are taken with respect to subsystems of the parties within each hyperdege.

Similarly, the partition-reduction matrix associated with the $n=4$ 2-uniform-complete hypergraph is displayed in Table \ref{example42}. Since the vertex set remains unchanged, the macroscopic quantities are the same as those of the $n=4$ 3-uniform-complete hypergraph. However, with the hyperedges reduced to 2-edges in this case, the microscopic quantities consist simply of the measures corresponding to the bipartition of each edge, as indicated at the top of the columns.
\begin{table*}[ht]
\centering
\begin{tabular}{|c|c|c|c|c|c|c|}
\hline
 & ${e_1^{(2)}}_{(A|B)}$ & ${e_2^{(2)}}_{(A|C)}$ & ${e_3^{(2)}}_{(A|D)}$ & ${e_4^{(2)}}_{(B|C)}$ & ${e_5^{(2)}}_{(B|D)}$ & ${e_6^{(2)}}_{(C|D)}$ \\
\hline
$E^{(2)}_{(ABC|D)}$ & $0$&$0$&$1$&$0$&$1$&$1$ \\
\hline
$E^{(2)}_{(ABD|C)}$ & $0$&$1$&$0$&$1$&$0$&$1$ \\
\hline
$E^{(2)}_{(AB|CD)}$ & $0$&$1$&$1$&$1$&$1$&$0$ \\
\hline
$E^{(3)}_{(AB|C|D)}$ & $0$&$1$&$1$&$1$&$1$&$1$ \\
\hline
$E^{(2)}_{(B|ACD)}$ & $1$&$0$&$0$&$1$&$1$&$0$ \\
\hline
$E^{(2)}_{(AC|BD)}$ & $1$&$0$&$1$&$1$&$0$&$1$ \\
\hline
$E^{(3)}_{(AC|B|D)}$ & $1$&$0$&$1$&$1$&$1$&$1$ \\
\hline
$E^{(2)}_{(AD|BC)}$ & $1$&$1$&$0$&$0$&$1$&$1$ \\
\hline
$E^{(2)}_{(A|BCD)}$ & $1$&$1$&$1$&$0$&$0$&$0$ \\
\hline
$E^{(3)}_{(BC|A|D)}$ & $1$&$1$&$1$&$0$&$1$&$1$ \\
\hline
$E^{(3)}_{(AD|B|C)}$ & $1$&$1$&$0$&$1$&$1$&$1$ \\
\hline
$E^{(3)}_{(BD|A|C)}$ & $1$&$1$&$1$&$1$&$0$&$1$ \\
\hline
$E^{(3)}_{(CD|A|B)}$ & $1$&$1$&$1$&$1$&$1$&$0$ \\
\hline
$E^{(4)}_{(A|B|C|D)}$ & $1$&$1$&$1$&$1$&$1$&$1$ \\
\hline
\end{tabular}
\caption{The partition-reduction matrix associated with the $n=4$ 3-uniform-complete hypergraph.}
\label{example42}
\end{table*}

The equalities of the $n=4$ 3-uniform-complete hypergraph Sperner states are:
\begin{align}
   (1)&\ E^{(2)}_{(A|BCD)}+E^{(2)}_{(B|ACD)}+
E^{(2)}_{(C|ABD)}+E^{(2)}_{(D|ABC)}-E^{(2)}_{(AB|CD)}-E^{(2)}_{(AC|BD)}
-E^{(2)}_{(AD|BC)=0},\\
(2)&\ E^{(4)}_{(A|B|C|D)}-\frac{1}{3}\big(E^{(3)}_{(AB|C|D)}+E^{(3)}_{(AC|B|D)}+E^{(3)}_{(AD|B|C)}+E^{(3)}_{(BC|A|D)}+E^{(3)}_{(BD|A|C)}+E^{(3)}_{(CD|A|B)}\big)\notag\\&+\frac{1}{3}\big(E^{(2)}_{(AB|CD)}+E^{(2)}_{(AC|BD)}+E^{(2)}_{(AD|BC)}\big)=0,\\
(3)&\ E^{(3)}_{(AD|B|C)}+E^{(3)}_{(BC|A|D)}-E^{(3)}_{(AB|C|D)}-E^{(3)}_{(CD|A|B)}+E^{(2)}_{(AB|CD)}-E^{(2)}_{(AD|BC)}=0,\\
(4)&\ E^{(3)}_{(AD|B|C)}+E^{(3)}_{(BC|A|D)}-E^{(3)}_{(AC|B|D)}-E^{(3)}_{(BD|A|C)}+E^{(2)}_{(AC|BD)}-E^{(2)}_{(AD|BC)}=0,
\end{align}
where (1) and (2) appeared in \cite{Iizuka:2025ioc}.

Comparing with the 3-uniform-complete hypergraph case, the Sperner states associated with the 2-uniform-complete hypergraph obey four more equalities as follows.
\begin{align}
    (5)\ &48E^{(4)}_{(A|B|C|D)}+10\big(E^{(3)}_{(AB|C|D)}+E^{(3)}_{(AC|B|D)}+E^{(3)}_{(AD|B|C)}+E^{(3)}_{(BC|A|D)}+E^{(3)}_{(BD|A|C)}+E^{(3)}_{(CD|A|B)}\big)\notag\\&-28\big(E^{(2)}_{(AB|CD)}+E^{(2)}_{(AC|BD)}+E^{(2)}_{(AD|BC)}\big)-21\big(E^{(2)}_{(A|BCD)}+E^{(2)}_{(B|ACD)}+E^{(2)}_{(C|ABD)}+E^{(2)}_{(D|ABC)}\big)=0,\\
    (6)\ &E^{(3)}_{(BD|A|C)}+E^{(3)}_{(CD|A|B)}-E^{(3)}_{(AB|C|D)}-E^{(3)}_{(AC|B|D)}-E^{(2)}_{(A|BCD)}+E^{(2)}_{(D|ABC)}=0,\\
    (7)\ &E^{(3)}_{(AD|B|C)}+E^{(3)}_{(BD|A|C)}-E^{(3)}_{(AC|B|D)}-E^{(3)}_{(BC|A|D)}-E^{(2)}_{(C|ABD)}+E^{(2)}_{(D|ABC)}=0,\\
    (8)\ &E^{(3)}_{(BC|A|D)}+E^{(3)}_{(BD|A|C)}-E^{(3)}_{(AC|B|D)}-E^{(3)}_{(AD|B|C)}-E^{(2)}_{(A|BCD)}+E^{(2)}_{(B|ACD)}=0,
\end{align}
where the first equality is symmetric with respect to the permutation among $A$, $B$, $C$, and $D$, and the latter three equalities are asymmetric.

\end{document}